\documentclass[english]{paper}
\usepackage[T1]{fontenc}
\usepackage[latin9]{inputenc}
\usepackage{babel}

\usepackage{amsthm}
\usepackage{amsmath}
\usepackage{amssymb}
\usepackage[authoryear]{natbib}
\usepackage[unicode=true, pdfusetitle,
 bookmarks=true,bookmarksnumbered=false,bookmarksopen=false,
 breaklinks=false,pdfborder={0 0 1},backref=false,colorlinks=false]
 {hyperref}

\makeatletter
\newcommand{\lyxaddress}[1]{
\par {\raggedright #1
\vspace{1.4em}
\noindent\par}
}
\theoremstyle{plain}
\newtheorem{thm}{Theorem}
  \theoremstyle{plain}
  \newtheorem{lem}[thm]{Lemma}
  \theoremstyle{plain}
  \newtheorem{prop}[thm]{Proposition}
  \theoremstyle{remark}
  \newtheorem{rem}[thm]{Remark}

\makeatother

\begin{document}

\title{Kolmogorov-Loveland Sets and Advice Complexity Classes}

\author{Thomas Hugel}

\maketitle

\lyxaddress{LIAFA - Université Paris 7 \& CNRS - case 7014\\
F-75205 Paris Cedex 13}
\begin{abstract}
Loveland complexity \citet{Loveland1969} is a variant of Kolmogorov
complexity, where it is asked to output separately the bits of the
desired string, instead of the string itself. Similarly to the resource-bounded
Kolmogorov sets we define Loveland sets.

We highlight a structural connection between resource-bounded Loveland
sets and some advice complexity classes.

This structural connection enables us to map to advice complexity
classes some properties of Kolmogorov sets first noticed by Hartmanis
\citet{Hartmanis1983} and thoroughly investigated in Longpré's thesis
\citet{Longpre1986}:
\begin{enumerate}
\item Non-inclusion properties of Loveland sets result in hierarchy properties
on the corresponding advice complexity classes;{\small \par}
\item Immunity properties of Loveland sets result in the non-existence of
natural proofs between the corresponding advice complexity classes,
in the sense of Razborov \& Rudich \citet{Razborov1997}.{\small \par}
\end{enumerate}
\end{abstract}

\section{Introduction}

Kolmogorov complexity is a measure of algorithmic randomness. If we
consider words of length $n$ over the boolean alphabet $\Sigma=\left\{ 0,1\right\} $,
then $KS\left(x\right)$ is defined to be the length of the smallest
input enabling a given universal Turing machine $U$ to output $x$.
It is possible to {}``hide'' $x$ in the input, so $KS\left(x\right)\leq l\left(x\right)+O\left(1\right)$
(where $l\left(x\right)$ is the length of $x$). Since there are
$2^{n}-1$ words in $\Sigma^{\leq n-1}$ but $2^{n}$ words in $\Sigma^{n}$,
by the pigeonhole principle, at least one word of length $n$ cannot
be output with a smaller input. Such a word is called \emph{incompressible}.
Numerous applications of this result, known as \emph{the Incompressibility
Method} are given in the textbook of Li \& Vitanyi \citet{Li2008}.
The same argument shows that almost all words are almost incompressible.

The Shannon-Lupanov theorem \citet{Shannon1949,Lupanov1958} states
that any boolean function on $n$ variables can be computed by a circuit
of size at most $\frac{2^{n}}{n}\left(1+o\left(1\right)\right)$.
This bound is tight; moreover almost all boolean functions have almost
this complexity. Trakthenbrot \citet{Trakhtenbrot1984} drew a parallel
between boolean functions having a circuit complexity of at least
$\left(1-\varepsilon\right)\frac{2^{n}}{n}$ and words such that $KS\left(x\right)\geq\left(1-\varepsilon\right)l\left(x\right)$.
Such an analogy suggests that there may be a connection between hard
boolean functions and incompressible words.

Indeed, Karp \& Lipton noticed that languages in $P/\mathrm{log}$
could be computed by {}``small circuits with easy descriptions''.
Looking for some kind of reciprocal, Hermo \& Mayordomo \citet{Hermo1994}
gave a nice characterization of the advice complexity class $P/\mathrm{log}$
in terms of the resource-bounded Kolmogorov complexity of its circuits.
Assuming that $t\left(n\right)$ and $s\left(n\right)$ are constructible
functions, a sequence of words $\left(x_{n}\right)$ of length $n$
belongs to the \emph{Kolmogorov set} $KS\left[f\left(n\right),t\left(n\right),s\left(n\right)|y\left(n\right)\right]$
iff for each integer $n$, $x_{n}$ is computable by the universal
Turing machine $U$ with an input of length at most $f\left(n\right)$
in time $t\left(n\right)$, in space $s\left(n\right)$ and with knowledge
of some word $y\left(n\right)$. Using the $P$-completeness of the
Circuit Value Problem, they proved that $P/\mathrm{log}$ is the set
of languages decidable by circuits belonging to $KS\left[O\left(\log n\right),\mathrm{poly}\left(n\right),+\infty|n\right]$
or, equivalently, by circuits belonging to $KS\left[O\left(\log n\right),+\infty,O\left(\log n\right)|n\right]$.

Building on Hartmanis \citet{Hartmanis1983}, Longpré \citet{Longpre1986}
stated many structural properties of the above Kolmogorov sets. However,
the characterization of Hermo \& Mayordomo does not make it possible
to map the properties of the Kolmogorov sets to the advice complexity
classes, since several different circuits may compute the same function
and not all words encode circuits. To do so one needs a more narrow
approach.

The purpose of the present work is to provide a characterization of
some advice complexity classes by Kolmogorov sets enabling to map
some of the properties of the latter to the former:
\begin{enumerate}
\item Non-inclusion properties of Kolmogorov sets result in hierarchies
properties on the corresponding advice complexity classes;
\item Immunity properties of Kolmogorov sets result in the non-existence
of natural proofs between the corresponding advice complexity classes,
in the sense of Razborov \& Rudich \citet{Razborov1997}.
\end{enumerate}
In a breakthrough paper \citet{Razborov1997}, Razborov \& Rudich
identified some properties shared by all known proofs of lower bounds
on the non-monotone circuit complexity of some individual boolean
functions, which they called \emph{natural properties}, and showed
that such properties would not be sufficient to derive superpolynomial
lower bounds for circuits (with some hardness assumption).

We believe that Kolmogorov sets are a convenient substitute for diagonalization,
since they are more compact and explicit. In particular, the results
presented in this paper may be viewed as resulting from diagonalization.

\section{Connection between Kolmogorov Sets and Advice Complexity Classes}

\subsection{Framework\label{sub:Framework}}

\paragraph{Kolmogorov-Loveland Complexity.}

To make the connection as clear as possible we do not use simple Kolmogorov
complexity $KS$ but a variant introduced by Loveland \citet{Loveland1969}
called \emph{decision Kolmogorov complexity} ($KD$) in the classification
of Uspensky \& Shen \citet{Uspensky1996}, and also referred to as
\emph{uniform complexity} in \citet{An-Introduction-to-Kolmogorov-Complexity-and-Its-Applications::1997::LiVitanyi}.
In this variant, the universal Turing machine $U$ is not required
to output the word $x$ but to give $x\left[i\right]$ (the $i^{\mathrm{th}}$
bit of $x$) upon request. The input of $U$ consists thus in a self-delimited
tuple $\left\langle p,i,y\right\rangle $ where $p$ is the actual
input, $i$ is the number of the asked bit of $x$, and $y$ is a
word given for free as an auxiliary input. We take as usual $\left\langle a,b\right\rangle =1^{l\left(a\right)}0ab$
and $\left\langle a,b,c\right\rangle =\left\langle \left\langle a,b\right\rangle ,c\right\rangle $.

\paragraph{Kolmogorov-Loveland Sets.}

In the setting of Kolmogorov-Loveland complexity, the required output
is a single bit, so it does not make sense to measure time and space
with respect to the output. Instead we are going to measure them with
respect to the input. This will allow a simpler connection with advice
complexity classes. So we say that a sequence of words $x_{n}$ of
length $n$ belongs to $KD\left[f\left(n\right),t\left(u\right),s\left(u\right)|y\left(n\right)\right]$
iff for each integer $n$ there exists a program $p\left(n\right)$
of length at most $f\left(n\right)$ such that for all $i$ between
$0$ and $n-1$, $U\left(\left\langle p\left(n\right),i,y\left(n\right)\right\rangle \right)=x\left[i\right]$
and the computation is done in time at most $t\left(u\right)$ and
in space at most $s\left(u\right)$, where $u=l\left(\left\langle p\left(n\right),y\left(n\right)\right\rangle \right)$.
The choice of $u$ may seem rather strange at first glance, but we
choose this definition so that the connection with advice complexity
classes may be as simple as possible. Moreover bounding resources
with respect to input size makes this complexity \emph{prefix-monotonic},
in the sense that the program $p$ and the word $y$ used for a word
$x$ in time $t\left(u\right)$ and in space $s\left(u\right)$ can
also be used for all of $x$'s prefixes.

\paragraph{Advice Complexity Classes.}

The advice complexity classes we are going to consider are of the
form $DTISP\left(t\left(u\right),s\left(u\right)\right)/f\left(n\right)$.
Given two fully time and space constructible functions $t\left(u\right)$
and $s\left(u\right)$, a language $L$ over the alphabet $\Sigma=\left\{ 0,1\right\} $
is in $DTISP\left(t\left(u\right),s\left(u\right)\right)$ iff there
exists a Turing machine deciding it in time $O\left(t\left(u\right)\right)$
and in space $O\left(s\left(u\right)\right)$. Now the advice complexity
classes are defined as follows: let $\mathcal{C}$ be a complexity
class. $L\in\mathcal{C}/f\left(n\right)$ iff there exists a language
$L'$ in $\mathcal{C}$ such that for all integer $n$, there exists
a word $w_{n}$ (the advice) of length at most $f\left(n\right)$
such that for any word $x$ of length $n,$ $x\in L\iff\left\langle w_{n},x\right\rangle \in L'$.

\paragraph{Characteristic Words Sequence of a Language.}

Given a language $L$, we consider its \emph{characteristic words
sequence}, defined as follows. First we recall that a word $x$ of
length $n$ over the alphabet $\Sigma$ can be considered (in the
binary numeral system) as an integer $\mathrm{int}\left(x\right)$
between $0$ and $2^{n}-1$. We denote the inverse function from integers
to words as $\mathrm{word}_{n}\left(i\right)$. Note that $\mathrm{word}_{n}\left(i\right)$
will have length $n$ and may begin by a sequence of zeros. In particular
the word $0^{n}$ corresponds to the integer $0$ and the word $1^{n}$
to the integer $2^{n}-1$ (keep this in mind).  Now the characteristic
word of the language $L$ for length $n$ is a word of length $N=2^{n}$
defined by $L_{n}\left[i\right]=\boldsymbol{1}_{\mathrm{word}_{n}\left(i\right)\in L}$
(i.e.~$1$ if $\mathrm{word}_{n}\left(i\right)\in L$ and $0$ otherwise).

\subsection{Connection Lemma}
\begin{lem}
\label{lem:Connection}Let $f\left(n\right)$, $t\left(u\right)$
and $s\left(u\right)$ be integer functions such that $t\left(u\right)$
and $s\left(u\right)$ are non-decreasing and fully time and space
constructible, $t\left(u\right)\geq u$ and $s\left(u\right)\geq\log u$:
\begin{enumerate}
\item if $L\in DTISP\left(t\left(u\right),s\left(u\right)\right)/f\left(n\right)$,
then $\left(L_{n}\right)\in KD\left[f\left(\log N\right)+O\left(1\right),O\left(t\left(u\right)\log t\left(u\right)\right),O\left(s\left(u\right)\right)|1^{\log N}\right]$
\item if $\left(L_{n}\right)\in KD\left[f\left(\log N\right),t\left(u\right),s\left(u\right)|1^{\log N}\right]$,
then $L\in DTISP\left(t\left(u\right),s\left(u\right)\right)/f\left(n\right)$.
\end{enumerate}
\end{lem}
\begin{proof}
The main idea is to switch from integers to strings of a given length
and vice-versa.
\begin{enumerate}
\item if $L\in DTISP\left(t\left(u\right),s\left(u\right)\right)/f\left(n\right)$,
then there exists a language $L'$ decided by a Turing machine $M'$
in time $O\left(t\left(u\right)\right)$ and in space $O\left(s\left(u\right)\right)$
such that for all integer $n$, there exists an advice word $w_{n}$
of length at most $f\left(n\right)$ such that $x\in L\iff\left\langle w_{n},x\right\rangle \in L'$.
Now consider the Turing machine $M$ which on input $\left\langle p,i,y\right\rangle $
computes $n=l\left(y\right)$ and $x=0^{n-l\left(i\right)}i$, and
simulates $M'$ on $\left\langle p,x\right\rangle $. By construction,
$M\left\langle w_{n},i,y\right\rangle =L_{n}\left[i\right]$ and $M$
runs in time $O\left(u+t\left(u\right)\right)=O\left(t\left(u\right)\right)$
and in space $O\left(\log u+s\left(u\right)\right)=O\left(s\left(u\right)\right)$.
Now $M$ can be simulated by the universal Turing machine $U$ (which
has a fixed number of tapes) in time $O\left(t\left(u\right)\log t\left(u\right)\right)$
and in space $O\left(s\left(u\right)\right)$ thanks to the simulation
method of Hennie \& Stearns \citet{Two-Tape-Simulation-of-Multitape-Turing-Machines::1966::HennieStearns}.
\item if $\left(L_{n}\right)\in KD\left[f\left(\log N\right),t\left(u\right),s\left(u\right)|1^{\log N}\right]$,
then for all integer $n$, there exists a program $p_{n}$ of length
at most $f\left(n\right)$ such that for all integer $i$ between
$1$ and $N$, $U\left\langle p_{n},i,1^{n}\right\rangle =L_{n}\left[i\right]$.
Now $p_{n}$ can be used as an advice, as follows: consider the Turing
machine $M$ which on input $\left\langle w,x\right\rangle $ computes
$n=l\left(x\right)$, $i=x$ without the inital zeros, and simulates
$U$ on $\left\langle w,i,1^{n}\right\rangle $. Then $M\left\langle p_{n},x\right\rangle =\boldsymbol{1}_{x\in L}$
and this computation is done in time $O\left(u+t\left(u\right)\right)=O\left(t\left(u\right)\right)$
and in space $O\left(\log u+s\left(u\right)\right)=O\left(s\left(u\right)\right)$.
\end{enumerate}
\end{proof}

\section{Transfer of Properties}

We are going to use the above Connection Lemma to transfer properties
of Loveland sets to the corresponding advice complexity classes:
\begin{enumerate}
\item Non-inclusion properties of Loveland sets result in hierarchy properties
on the corresponding advice complexity classes;
\item Immunity properties of Loveland sets result in the non-existence of
natural proofs between the corresponding advice complexity classes.
\end{enumerate}
So we shall first establish the properties of Loveland sets. They
are very similar to the properties established by Longpré in his thesis,
but we must revisit them because we work with Loveland complexity
instead of standard Kolmogorov complexity, and moreover we measure
the resources with respect to the input instead of the output.

\subsection{Non-Inclusions and Hierarchies}

\subsubsection{Sensitivity to Advice Length.}

The following proposition is analogous to corollary 3.2 of Longpré's
thesis \citet{Longpre1986}. Here the proof is simpler because $KD$
is prefix-monotonic.
\begin{prop}
\label{fac:KD-advice-length}If $f\left(n\right)<n$, then there exists
a constant $c$ such that for all $y\left(n\right)$ we have \[
KD\left[f\left(n\right)+c,cu\log u,c\log u|y\left(n\right)\right]\nsubseteq KD\left[f\left(n\right),+\infty,+\infty|y\left(n\right)\right]\enskip.\]
\end{prop}
\begin{proof}
By counting, there exists some $x$ of length $f\left(n\right)+1$
incompressible with respect to $y\left(n\right)$; by prefix-monotonicity,
$z=x0^{n-f\left(n\right)-1}\notin KD\left[f\left(n\right),+\infty,+\infty|y\left(n\right)\right]$.
Now the program which on input $\left\langle x,i,y\right\rangle $
prints $x\left[i\right]$ if $i<l\left(x\right)$ and $0$ otherwise,
has length $f\left(n\right)+O\left(1\right)$, works in time $O\left(l\left(x\right)+l\left(y\right)\right)$
and in space $O\left(\log l\left(x\right)\right)$. As above, this
program can be simulated by $U$ in time $O\left(u\log u\right)$
and in space $O\left(\log u\right)$, where $u=l\left(\left\langle x,y\right\rangle \right)$.
So there exists a constant $c$ such that $z\in KD\left[f\left(n\right)+c,cu\log u,c\log u|y\left(n\right)\right]$.
\end{proof}
Now using this proposition together with our Connection Lemma yields
the following result, which was already present in \citet{Hermo1994}
in a similar form. We denote by $REC$ the class of recursive languages.

\begin{thm}
Let $f\left(n\right)$ and $g\left(n\right)$ be two integer functions
such that $f\left(n\right)+\frac{g}{2}\left(n\right)<2^{n}$ and $g\left(n\right)$
is non-decreasing and unbounded. Then \[
DTISP\left(u\log u,\log u\right)/\left(f\left(n\right)+g\left(n\right)\right)\nsubseteq REC/f\left(n\right)\enskip.\]
\end{thm}
\begin{proof}
By proposition \ref{fac:KD-advice-length} and the fact that $g\left(n\right)$
is non-decreasing and unbounded, we consider a sequence of words $L_{n}$
of length $N=2^{n}$ belonging to $KD\left[\left(f+g\right)\left(\log N\right),cu\log u,c\log u|1^{\log N}\right]\backslash KD\left[\left(f+\frac{g}{2}\right)\left(\log N\right),+\infty,+\infty|1^{\log N}\right]$
for some $c$ and for $n$ large enough. Now by the above Connection
Lemma, the corresponding language $L$ is in $DTISP\left(u\log u,\log u\right)/\left(f\left(n\right)+g\left(n\right)\right)$.
Suppose by contradiction that $L\in REC/f\left(n\right)$. By the
Connection Lemma, there exists a constant $c'$ such that $\left(L_{n}\right)$
is in $KD\left[f\left(\log N\right)+c',+\infty,+\infty|1^{\log N}\right]$.
For $n$ large enough, $g\left(n\right)>2c'$ and then $\left(L_{n}\right)$
is in $KD\left[\left(f+\frac{g}{2}\right)\left(\log N\right),+\infty,+\infty|1^{\log N}\right]$,
a contradiction.
\end{proof}

\subsubsection{Sensitivity to Time and Space.}

We give a proposition analogous to theorems 4.3 and 4.4 of Longpré's
thesis \citet{Longpre1986}. The main difference here is that we consider
that resources are bounded with respect to the input rather than the
output.
\begin{prop}
\label{fac:Time-Space}Let $f\left(n\right)$, $t\left(u\right)$
and $s\left(u\right)$ be integer non-decreasing and constructible
functions such that $f\left(n\right)<n$, $t\left(u\right)\geq u$
and $s\left(u\right)\geq\log u$. Then there exists some constant
$c$ such that if $t'\left(u\right)\geq c2^{f\left(2^{u}\right)}f\left(2^{u}\right)t\left(2f\left(2^{u}\right)+u\right)\left(f\left(2^{u}\right)+\log t\left(2f\left(2^{u}\right)+u\right)\right)$
and $s'\left(u\right)\geq c\left(2^{f\left(2^{u}\right)}f\left(2^{u}\right)+s\left(2f\left(2^{u}\right)+u\right)\right)$,
for $n$ large enough, \[
KD\left[c,t'\left(u\right),s'\left(u\right)|n-1\right]\nsubseteq KD\left[f\left(n\right),t\left(u\right),s\left(u\right)|n-1\right]\enskip.\]
\end{prop}
\begin{proof}
By counting, there must exist some $x$ of length $f\left(n\right)+1$
incompressible with respect to $n-1$ in time $t\left(u\right)$ and
in space $s\left(u\right)$; by prefix-monotonicity, $z=x0^{n-f\left(n\right)-1}\notin KD\left[f\left(n\right),t\left(u\right),s\left(u\right)|n-1\right]$.
Finding the smallest such $x$ in the lexicographic order can be performed
by exhaustive search by running all programs of length at most $f\left(n\right)$
in time $t\left(u\right)$ and in space $s\left(u\right)$. Here we
face a trade-off between time and space: if we choose to store all
generated strings to avoid recomputations, this increases the required
space; otherwise we may iterate the exhaustive search for each string
of length $f\left(n\right)+1$ until we find the desired one. The
first option takes:
\begin{itemize}
\item an overall advice of length $O\left(1\right)$ (since $n-1$ is given
for free);
\item an overall time of $O\left(2^{f\left(n\right)}f\left(n\right)t\left(2f\left(n\right)+\log n+2\right)\right)$,
since $l\left(\left\langle a,b\right\rangle \right)=2l\left(a\right)+l\left(b\right)+1$
and $l\left(n-1\right)\leq\log n+1$; now $u=l\left(\left\langle p,n-1\right\rangle \right)>\log n+1$.
So the time bound is $O\left(2^{f\left(2^{u}\right)}f\left(2^{u}\right)t\left(2f\left(2^{u}\right)+u\right)\right)$.
Again there is an extra logarithmic factor due to the simulation by
our fixed machine $U$;
\item an overall space of $O\left(2^{f\left(n\right)}f\left(n\right)+s\left(2f\left(n\right)+\log n+2\right)\right)$,
i.e. $O\left(2^{f\left(2^{u}\right)}f\left(2^{u}\right)+s\left(2f\left(2^{u}\right)+u\right)\right)$.
\end{itemize}
\end{proof}
Now using this proposition together with our Connection Lemma yields
the following result.
\begin{thm}
Let $f\left(n\right)$, $g\left(n\right)$, $t\left(u\right)$ and
$s\left(u\right)$ be integer non-decreasing and constructible functions
such that $f\left(n\right)+g\left(n\right)<2^{n}$, $g$ is unbounded,
$t\left(u\right)\geq u$ and $s\left(u\right)\geq\log u$. Let $t''\left(u\right)$
and $s''\left(u\right)$ be such that $t''\left(u\right)=\omega\left(t\left(u\right)\log t\left(u\right)\right)$
and $s''\left(u\right)=\omega\left(s\left(u\right)\right)$. Then
there exists some constant $c$ such that if $t'\left(u\right)\geq c2^{\left(f+g\right)\left(u\right)}\left(f+g\right)\left(u\right)t''\left(2\left(f+g\right)\left(u\right)+u\right)\left(\left(f+g\right)\left(u\right)+\log t''\left(2\left(f+g\right)\left(u\right)+u\right)\right)$
and $s'\left(u\right)\geq c\left(2^{\left(f+g\right)\left(u\right)}\left(f+g\right)\left(u\right)+s''\left(2\left(f+g\right)\left(u\right)+u\right)\right)$,
then \[
DTISP\left(t'\left(u\right),s'\left(u\right)\right)/c\nsubseteq DTISP\left(t\left(u\right),s\left(u\right)\right)/f\left(n\right)\enskip.\]
\end{thm}
\begin{proof}
By proposition \ref{fac:Time-Space}, we consider a sequence of words
$L_{n}$ of length $N=2^{n}$ belonging to $KD\left[c,t'\left(u\right),s'\left(u\right)|1^{\log N}\right]\backslash KD\left[\left(f+g\right)\left(\log N\right),t''\left(u\right),s''\left(u\right)|1^{\log N}\right]$
for some $c$ and for $n$ large enough. Now by the Connection Lemma,
the corresponding language $L$ is in $DTISP\left(t'\left(u\right),s'\left(u\right)\right)/c$.
Suppose by contradiction that $L\in DTISP\left(t\left(u\right),s\left(u\right)\right)/f\left(n\right)$.
By the Connection Lemma, there exists a constant $c'$ such that $\left(L_{n}\right)$
is in $KD\left[f\left(\log N\right)+c',c't\left(u\right)\log t\left(u\right),c's\left(u\right)|1^{\log N}\right]$.
For $n$ large enough, $\left(f+g\right)\left(\log N\right)\geq f\left(\log N\right)+c'$,
$t''\left(u\right)\geq c't\left(u\right)\log t\left(u\right)$ and
$s''\left(u\right)\geq c's\left(u\right)$, so $\left(L_{n}\right)$
is in $KD\left[\left(f+g\right)\left(\log N\right),t''\left(u\right),s''\left(u\right)|1^{\log N}\right]$,
a contradiction.
\end{proof}

\subsection{Immunity and Natural Proofs}

\subsubsection{Immunity of Kolmogorov Sets.}

Immunity is an indication that a language is algorithmically very
complex, in the sense that given a complexity class $\mathcal{C}$,
a language $L$ is called $\mathcal{C}$-\emph{immune} iff $L$ is
infinite and does not have any infinite subset belonging to $\mathcal{C}$.
To this we add the notion of density: a language $L$ has \emph{partial
density} $\delta$ iff there exist infinitely many $n$'s such that
$L$ contains at least $\delta\left(n\right)2^{n}$ words of length
$n$. Thus we generalize the notion of immunity using density: we
say that a language $L$ is \emph{$\mathcal{C}$-immune for partial
density $\delta$} iff $L$ is infinite and does not have any infinite
subset of partial density at least $\delta$ belonging to $\mathcal{C}$.

As noted by Hartmanis \citet{Hartmanis1983} and further developed
by Longpré \citet{Longpre1986}, the complements of Kolmogorov sets
are immune. Longpré's results for immunity (theorems 3.7, 3.8 and
4.13 of \citet{Longpre1986}) concern classical complexity classes
and global density. Here we deal with advice complexity classes and
partial density, as follows.
\begin{prop}
\label{prop:Immunity}Let $f\left(n\right)$, $g\left(n\right)$,
$t\left(u\right)$ and $s\left(u\right)$ be integer non-decreasing
and constructible functions, such that $f\left(n\right)<n$, $g\left(n\right)$
is unbounded and $g\left(n\right)<f\left(n\right)$. Let $\delta\left(n\right)$
be a function to the real interval $\left[0,1\right]$ and $\rho\left(n\right)=\left(1-\delta\left(n\right)\right)2^{n}+1$.
If $t'\left(u\right)\geq u$ and $s'\left(u\right)\geq u$ are non-decreasing,
$t'\left(u\right)=o\left(\frac{t\left(\log u\right)}{\rho\left(u\right)\log\left(\rho\left(u\right)t\left(\log u\right)\right)}\right)$
and $s'\left(u\right)=o\left(s\left(\log u\right)\right)$, then

$\Sigma^{*}\backslash\bigcup_{n\in\boldsymbol{N}}KD\left[f\left(n\right),t\left(u\right),s\left(u\right)|n-1\right]$
is $DTISP\left(t'\left(u\right),s'\left(u\right)\right)/\left(f-g\right)\left(n\right)$-immune
for partial density $\delta$.\end{prop}
\begin{proof}
Let us consider any infinite language $A\in DTISP\left(t'\left(u\right),s'\left(u\right)\right)/\left(f-g\right)\left(n\right)$
with partial density $\delta$. We argue that for $n$ large enough,
the lexicographically smallest word of length $n$ belonging to $A$
is in $KD\left[f\left(n\right),t\left(u\right),s\left(u\right)|n-1\right]$.
Indeed there exists a Turing machine $M$ working in time $t'\left(u\right)$
and in space $s'\left(u\right)$, and a sequence of advice $\left(w_{n}\right)$
of length $l\left(w_{n}\right)\leq\left(f-g\right)\left(n\right)$
such that for all $x\in\Sigma^{n}$, $x\in A\iff M\left\langle w_{n},x\right\rangle =1$.
Thus it suffices to simulate $M$ with advice $w_{n}$ on all $x$'s
of length $n$ in the lexicographic order.
\begin{itemize}
\item This can be done by a program of length $\left(f-g\right)\left(n\right)+O\left(1\right)$
(since $n-1$ is given for free).
\item For an $n$ such that $A$ contains at least $\delta\left(n\right)2^{n}$
words of length $n$, there are at most $\left(1-\delta\left(n\right)\right)2^{n}+1=\rho\left(n\right)$
steps of simulation, each step requiring a time $t'\left(v\right)$
where $v=l\left(\left\langle w_{n},1^{n}\right\rangle \right)\leq2l\left(w_{n}\right)+n+1\leq2^{2l\left(w_{n}\right)+\log\left(n-1\right)+1}=2^{l\left(\left\langle w_{n},n-1\right\rangle \right)}=2^{u}$.
The simulation by our universal Turing machine $U$ can thus be performed
in time $O\left(\left(\rho t'\log\left(\rho t'\right)\right)\left(2^{u}\right)\right)=o\left(\frac{t\left(u\right)}{\log\left(\rho\left(2^{u}\right)t\left(u\right)\right)}\log\frac{t\left(u\right)}{\log\left(\rho\left(2^{u}\right)t\left(u\right)\right)}\right)$.
Since $\frac{t\left(u\right)}{\log\left(\rho\left(2^{u}\right)t\left(u\right)\right)}\leq t\left(u\right)\leq\rho\left(2^{u}\right)t\left(u\right)$,
this is $o\left(t\left(u\right)\right)$.
\item The above simulation can be performed in space $O\left(v+s'\left(v\right)\right)=O\left(s'\left(2^{u}\right)\right)=o\left(s\left(u\right)\right)$.
\end{itemize}
\end{proof}

\subsubsection{Non-Existence of Natural Proofs among Advice Complexity Classes.}

We first recall the definitions of \citet{Razborov1997} (section
2.2). A \emph{combinatorial property} is a set of boolean functions.
Each of the $2^{2^{n}}$ boolean functions on $n$-bit inputs can
be described by a binary word of length $2^{n}$ (which in turn can
be seen as the characteristic word of a language, see section \ref{sub:Framework}
above). Thus a combinatorial property can be seen as a language with
words of length powers of $2$. The question whether a given boolean
function belongs to a combinatorial property is an algorithmic problem
which requires some time and space depending on the length $2^{n}$
of the boolean functions. Thus it is possible to group together combinatorial
properties with respect to this algorithmic complexity, and such sets
of combinatorial properties are some kinds of \emph{complexity classes}.
These complexity classes should not be confused with the complexity
classes of the boolean functions themselves, i.e. the time and space
(depending on $n$) required to compute the boolean functions on $n$-bit
inputs!

Given a complexity class $\mathcal{C}$, a combinatorial property
$\Gamma$ is called \emph{$\mathcal{C}$-natural for partial density
$\delta$ }iff there exists $\Xi\subseteq\Gamma$ such that:
\begin{description}
\item [{constructibility:}] $\Xi\in\mathcal{C}$
\item [{largeness:}] $\Xi$ is infinite and has partial density $\delta\circ\log$
(since the words in $\Xi$ have lengths of the form $2^{n}$)
\end{description}
In fact the density considered in \citet{Razborov1997} is global,
but we refine it to use partial density. So, what does a non-natural
property look like? It is a property without any large constructible
sub-property. This looks very much like the aforementioned notion
of immunity. Indeed:
\begin{rem}
\label{rem:non-natural=00003Dimmune}Let $\Gamma$ be a combinatorial
property, $\mathcal{C}$ a complexity class and $\delta:\boldsymbol{N}\to\left[0,1\right]$.
Then $\Gamma$ is not $\mathcal{C}$-natural for partial density $\delta$
iff $\Gamma$ is $\mathcal{C}$-immune for partial density $\delta\circ\log$.
\end{rem}
Now what is the use of a combinatorial property? Given a complexity
class $\mathcal{D}$, an infinite combinatorial property $\Gamma$
is called \emph{useful} against $\mathcal{D}$ iff
\begin{description}
\item [{usefulness:}] given a sequence $\left(L_{n}\right)$ of characteristic
words, if $L_{n}\in\Gamma$ infinitely often then $L\notin\mathcal{D}$.
\end{description}
Why is it called {}``usefulness''? Because in order to prove that
$L\notin\mathcal{D}$, it is enough to prove that $L_{n}\in\Gamma$
infinitely often. So Razborov \& Rudich manage to prove that for various
circuit complexity classes $\mathcal{C}$ and $\mathcal{D}$ there
are no $\mathcal{C}$-natural properties against $\mathcal{D}$. Using
proposition \ref{prop:Immunity} together with our Connection lemma,
we prove the following result:
\begin{thm}
Let $f\left(n\right)$, $g\left(n\right)$, $t\left(u\right)$ and
$s\left(u\right)$ be integer non-decreasing and constructible functions,
such that $f\left(n\right)<2^{n}$, $g$ is unbounded and $g\left(n\right)<f\left(n\right)$.
Let $\delta\left(n\right)$ be a function to the real interval $\left[0,1\right]$
and $\rho\left(n\right)=\left(1-\delta\left(\log n\right)\right)2^{n}+1$.
If $t'\left(u\right)\geq u$ and $s'\left(u\right)\geq u$ are non-decreasing,
$t'\left(u\right)=o\left(\frac{t\left(\log u\right)}{\rho\left(u\right)\log\left(\rho\left(u\right)t\left(\log u\right)\right)}\right)$
and $s'\left(u\right)=o\left(s\left(\log u\right)\right)$, then\end{thm}
\begin{quotation}
\noindent there is no $DTISP\left(t'\left(u\right),s'\left(u\right)\right)/\left(f-g\right)\left(\log n\right)$-natural
property for partial density $\delta$ useful against $DTISP\left(t\left(u\right),s\left(u\right)\right)/f\left(n\right)$.\end{quotation}
\begin{proof}
Let $\Gamma$ be a $DTISP\left(t'\left(u\right),s'\left(u\right)\right)/\left(f-g\right)\left(\log n\right)$-natural
property for partial density $\delta$. It is important to notice
that all words in $\Gamma$ have lengths of the form $N=2^{n}$. There
exists $\Xi\subseteq\Gamma$ such that $\Xi\in DTISP\left(t'\left(u\right),s'\left(u\right)\right)/\left(f-g\right)\left(\log N\right)$
and $\Xi$ has partial density $\delta\circ\log$. Suppose by contradiction
that $\Delta=\Xi\cap\left(\bigcup_{N\in\left\{ 2^{n}\right\} _{n\in\boldsymbol{N}}}KD\left[f\left(\log N\right),t\left(u\right),s\left(u\right)|1^{\log N}\right]\right)$
is finite. Then $\Xi\backslash\bigcup_{N\in\left\{ 2^{n}\right\} _{n\in\boldsymbol{N}}}KD\left[f\left(\log N\right),t\left(u\right),s\left(u\right)|1^{\log N}\right]$
is still infinite with partial density $\delta\circ\log$ and is still
in $DTISP\left(t'\left(u\right),s'\left(u\right)\right)/\left(f-g\right)\left(\log N\right)$,
which contradicts proposition \ref{prop:Immunity}. Thus $\Delta$
is infinite, and there are infinitely many $n$'s such that we can
pick in this set an element $L_{n}$ of length $N=2^{n}$. Now by
the Connection Lemma, the language $L$ having these $L_{n}$'s as
characteristic words is in $DTISP\left(t\left(u\right),s\left(u\right)\right)/f\left(n\right)$.
Thus $\Gamma$ is not useful against $DTISP\left(t\left(u\right),s\left(u\right)\right)/f\left(n\right)$.
\end{proof}

\section{Conclusion}

We had to take resource bounds on Kolmogorov-Loveland complexity with
respect to the input in order to make a straightforward connection
with the advice complexity classes. This led to use both $n$ and
$u$ in the complexity classes of the form $DTISP\left(t\left(u\right),s\left(u\right)\right)/f\left(n\right)$,
and we had to make frequent and inelegant conversions between $n$
and $u$. To make things clearer, we suggest that in the advice complexity
classes, resource bounds on $\left\langle w,x\right\rangle $ should
be taken with respect to $x$ only. This would not change the main
classes $P/\mathrm{poly}$ and $P/\mathrm{log}$.

One may object that our results relativize. They do as do all diagonalization
results, since our results are some kinds of diagonalizations. That
is also the reason why we think that it would be surprising if one
could do better than exhaustive search in our simulations.

However we believe that it is important to exhibit separation results,
even if they are simple. Indeed we want to recall for example that
in the deep proof that $DLIN\neq NLIN$ of \citet{Paul1983}, the
only separation result invoked is a simple time hierarchy on alternating
Turing machines obtained by diagonalization.

\section*{Acknowledgments}

This research was done while the author was an intern student at LRI
(Orsay, France) under supervision of Sophie Laplante. The author is
very grateful to Sophie Laplante for many helpful discussions.

\bibliographystyle{elsart-harv}
\bibliography{biblio,library}

\begin{thebibliography}{13}
\expandafter\ifx\csname natexlab\endcsname\relax\def\natexlab#1{#1}\fi
\expandafter\ifx\csname url\endcsname\relax
  \def\url#1{\texttt{#1}}\fi
\expandafter\ifx\csname urlprefix\endcsname\relax\def\urlprefix{URL }\fi

\bibitem[{Hartmanis(1983)}]{Hartmanis1983}
Hartmanis, J., Nov. 1983. {Generalized Kolmogorov complexity and the structure
  of feasible computations}. In: 24th Annual Symposium on Foundations of
  Computer Science (FOCS 1983). IEEE, pp. 439--445.
\newline\urlprefix\url{http://ieeexplore.ieee.org/lpdocs/epic03/wrapper.htm?ar%
number=4568108}

\bibitem[{Hennie and
  Stearns(1966)}]{Two-Tape-Simulation-of-Multitape-Turing-Machines::1966::Henn%
ieStearns}
Hennie, F.~C., Stearns, R.~E., 1966. Two-tape simulation of multitape {T}uring
  machines. J. ACM 13~(4), 533--546.

\bibitem[{Hermo and Mayordomo(1994)}]{Hermo1994}
Hermo, M., Mayordomo, E., Jul. 1994. {A Note on polynomial-size circuits with
  low resource-bounded Kolmogorov complexity}. Mathematical Systems Theory
  27~(4), 347--356.
\newline\urlprefix\url{http://www.springerlink.com/index/10.1007/BF01192144}

\bibitem[{Li and
  Vit{\'a}nyi(1997)}]{An-Introduction-to-Kolmogorov-Complexity-and-Its-Applica%
tions::1997::LiVitanyi}
Li, M., Vit{\'a}nyi, P., 1997. An Introduction to {K}olmogorov Complexity and
  Its Applications, 2nd Edition. Springer-Verlag, New-York.
\newline\urlprefix\url{citeseer.ifi.unizh.ch/li97introduction.html}

\bibitem[{Li and Vit\'{a}nyi(2008)}]{Li2008}
Li, M., Vit\'{a}nyi, P., 2008. {An Introduction to Kolmogorov Complexity and
  its Applications}, 3rd Edition. Springer Verlag.

\bibitem[{Longpr\'{e}(1986)}]{Longpre1986}
Longpr\'{e}, L., 1986. {Resource Bounded Kolmogorov Complexity, A Link between
  Computational Complexity and Information Theory}. Ph.D. thesis, Cornell
  University.

\bibitem[{Loveland(1969)}]{Loveland1969}
Loveland, D.~W., Dec. 1969. {A Variant of the Kolmogorov Concept of
  Complexity}. Information and Control 15~(6), 510--526.
\newline\urlprefix\url{http://linkinghub.elsevier.com/retrieve/pii/S0019995869%
905385}

\bibitem[{Lupanov(1958)}]{Lupanov1958}
Lupanov, O.~B., 1958. {A Method for Synthesizing Circuits}. Izv. vysshykh
  uchebnykh zavedenii, Radiofizika 1, 120--140.

\bibitem[{Paul et~al.(1983)Paul, Pippenger, Szemeredi, and Trotter}]{Paul1983}
Paul, W.~J., Pippenger, N., Szemeredi, E., Trotter, W.~T., Nov. 1983. {On
  determinism versus non-determinism and related problems}. In: 24th Annual
  Symposium on Foundations of Computer Science (FOCS 1983). IEEE, pp. 429--438.
\newline\urlprefix\url{http://ieeexplore.ieee.org/lpdocs/epic03/wrapper.htm?ar%
number=4568107}

\bibitem[{Razborov and Rudich(1997)}]{Razborov1997}
Razborov, A.~A., Rudich, S., Aug. 1997. {Natural Proofs}. Journal of Computer
  and System Sciences 55~(1), 24--35.

\bibitem[{Shannon(1949)}]{Shannon1949}
Shannon, C.~E., 1949. {The synthesis of two-terminal switching circuits}. Bell
  System Technical Journal 28, 59--98.

\bibitem[{Trakhtenbrot(1984)}]{Trakhtenbrot1984}
Trakhtenbrot, B.~A., Oct. 1984. {A Survey of Russian Approaches to Perebor
  (Brute-Force Searches) Algorithms}. IEEE Annals of the History of Computing
  6~(4), 384--400.
\newline\urlprefix\url{http://ieeexplore.ieee.org/lpdocs/epic03/wrapper.htm?ar%
number=4640789}

\bibitem[{Uspensky and Shen(1996)}]{Uspensky1996}
Uspensky, V.~A., Shen, A., Jun. 1996. {Relations between varieties of
  Kolmogorov complexities}. Mathematical Systems Theory 29~(3), 271--292.
\newline\urlprefix\url{http://www.springerlink.com/index/10.1007/BF01201280}

\end{thebibliography}

\end{document}